\documentclass[preprint,11pt]{elsarticle}
\usepackage{graphics}
\usepackage{graphicx}
\usepackage{epsfig}
\usepackage{amssymb}
\usepackage{amsthm}
\usepackage{tikz}
\usepackage{url}
\usepackage{algorithm,myalgorithmic}
\usetikzlibrary{arrows,matrix,backgrounds}

\newcommand{\mc}{\mathcal}

\newcommand{\np}{{\mc{NP}}}
\newtheorem{theorem}{Theorem}
\newtheorem{lemma}{Lemma}

\newtheorem{corollary}{Corollary}
\newtheorem{rrule}{Reduction Rule}

\begin{document}

\begin{frontmatter}

\title{On the Complexity of Multi-Parameterized Cluster Editing\tnoteref{note1}}
\tnotetext[note1]{\noindent A preliminary version of a portion of
this paper appeared in the $7^{th}$ Annual International Conference on 
Combinatorial Optimization and Applications (COCOA 2013) \cite{Abukhzam13}}

\author{Faisal N. Abu-Khzam}
\ead{\noindent faisal.abukhzam@lau.edu.lb}

\address{Department of Computer Science and Mathematics\\
Lebanese American University\\
Beirut, Lebanon}

\author{}
\address{}
\thispagestyle{empty}

\begin{abstract}

The Cluster Editing problem seeks a transformation of a given undirected 
graph into a disjoint union of cliques via a minimum number of edge 
additions or 
deletions. 
A multi-parameterized version of the problem is studied, featuring  
a number of input parameters that bound the 
amount of both edge-additions and deletions
per single vertex, as well as the size of a clique-cluster.
We show that the problem remains $\np$-hard even when only one edge
can be deleted and at most two edges can be added per vertex.
However, the new formulation allows us to solve Cluster Editing (exactly)
in polynomial time when the number of 
edge-edit operations per vertex is smaller than half the minimum cluster size.
In other words, Correlation Clustering can be solved efficiently 
when the number of false positives/negatives per single data element
is expected to be small compared to the minimum cluster size.
As a byproduct, 
we obtain a kernelization algorithm that delivers
linear-size kernels when the two edge-edit bounds are small constants.

\end{abstract}


\end{frontmatter}

\section{Introduction}

Given a simple undirected graph $G = (V, E)$ and an integer $k > 0$, the 
Cluster Editing Problem asks whether $k$ or less edge additions or deletions 
can transform $G$ into a graph whose connected components are cliques.
Cluster Editing is $\np$-Complete \cite{KM86,Shamir04}, but 
it is fixed-parameter tractable with respect to the parameter $k$ 
\cite{Cai96,Gramm05} 
\footnote{We assume familiarity with the notions
of fixed-parameter tractability and kernelization 
algorithms \cite{CFKLMPPS15,DF13,FG06,Niedermeier06}.}.
The problem received considerable attention recently as can 
be seen from a long sequence of continuous algorithmic improvements (see
\cite{Bocker12, BBBT09, BBK11,Cao-Chen,Chen-Meng,Gramm05,Guo09}).
The current asymptotically fastest fixed-parameter algorithm runs in 
$O^{*}(1.618^k)$ time \cite{Bocker12}. Moreover, a kernel of order $2k$ was
obtained recently in \cite{Chen-Meng}. This means that an arbitrary Cluster
Editing instance can be reduced in polynomial-time into an equivalent instance 
where the number of vertices is at most $2k$. The number of edges in the reduced
instance can be quadratic in $k$. 

Cluster Editing can be viewed as a model for 
accurate unsupervised ``correlation clustering.''
In such context, edges to be deleted or added from a given instance are 
considered false positives or false negatives, respectively. Such 
errors could be 
small in some practical applications, and they tend to be even smaller per 
input object, or vertex. In fact, a single data element that is causing 
too many false positive/negatives might be considered as outlier.

We consider a parameterized version of Cluster Editing
where both the number of edges that can be deleted and the number of edges 
that can be added, per vertex, are bounded by input parameters.
We refer to these two bounds by {\em error parameters}. 
Similar parameterizations appeared in \cite{Heggernes10} and \cite{KU12}.
In \cite{Heggernes10}, 
two parameters $p$ and $q$ were used to bound (respectively)
(i) the number of edges that can be added between elements of the same 
cluster and (ii) the number of edges that can be deleted between a 
cluster and the rest of the graph.
In \cite{KU12}, the total number of edge-edit operations, 
per vertex, and the number of clusters in the target solution are used 
as additional parameters.
We shall see that setting separate bounds on the two parameters 
could affect the complexity as well as algorithmic solutions of the problem. 

We introduce another parameter that bounds, from below, the minimum
acceptable cluster size and present a polynomial time algorithm 
that solves 
Cluster Editing, exactly, whenever the sum of error parameters is small 
compared to the minimum cluster size (in the target solution). 
This condition could be of particular interest
in applications where the cluster size is expected to be large or
when error parameters per data element are not expected to be high. 
In this respect, 
the message conveyed by our work bears the same theme as another, rather
experimental, study of various clustering methods conducted in \cite{DLS12}
where it was suggested that Clustering is not as hard as claimed by 
corresponding $\np$-hardness proofs. 

We shall also study the complexity of the multi-parameterized version
of Cluster Editing when the error-parameters are small constants. In 
particular, we show that Cluster Editing remains $\np$-hard even
when at most one edge can be deleted and at most two edges can be 
added per vertex. Moreover, we show in this case 
that a simple reduction procedure yields a problem kernel 
whose total size is linear in the parameter $k$. Previously
known kernelization algorithms cannot be applied to the considered
multi-parameterized version and they deliver kernels whose order (number 
of vertices only) is linear in $k$.

The paper is organized as follows: section 2 presents some preliminaries;
in section 3 we study the complexity of Cluster Editing when parameterized
by the error-parameters; section 4 is devoted to a general reduction procedure;
the consequent complexity results are presented in sections 5,
and section 6 concludes with a summary.

\section{Preliminaries}

We adopt common graph theoretic terminologies, such as neighborhood, vertex
degree, 
adjacency, etc. The term non-edge is used to designate a pair of 
non-adjacent vertices. Given a graph $G=(V,E)$, and a set
$S\subset V$, the subgraph induced by $S$ is denoted by $G[S]$. 
A clique in a graph is a subgraph induced by a set of pair-wise adjacent 
vertices. An edge-editing operation is either a deletion or an addition of 
an edge. We shall use the term {\it cluster graph} to denote a transitive 
undirected graph, which consists of a disjoint union of cliques, as 
connected components.

For a given graph $G$ and parameter $k$, the Parameterized Cluster Editing 
problem asks whether $G$ can be transformed into a cluster graph via $k$ or 
less edge-editing operations. 
In this paper, we consider a multi-parameterized 
version of this problem that assumes a set of parameters (independent of the
input). 
We shall use the list of parameters in the name of the problem. For example
$(a,d,k)$-Cluster Editing is formally defined as follows.

\newpage

\noindent
\underline{{\bf $(a,d,k)$-Cluster Editing}:}

\vspace{3pt}

\noindent
{\bf Input:} A graph $G$, parameters $k, a, d$, and two functions
$\alpha: V(G)\rightarrow \{0,1, \cdots, a\}$,
$\delta: V(G)\rightarrow \{0,1, \cdots, d\}$.

\noindent
{\bf Question:} Can $G$ be transformed into a disjoint union of 
cliques by at most $k$ edge-edit operations such that:
 
\noindent
for each vertex $v\in V(G)$, the number of added (deleted) edges incident 
on $v$ is at most $\alpha(v)$ ($\delta(v)$ respectively)?

\vspace{10pt}

\noindent
We shall further use some special terminology to better-present
our results. The
expression {\it solution graph} may be used instead of cluster graph, when 
dealing with a specific input instance. Edges that are not allowed to be 
in the cluster graph are called {\it forbidden} edges, 
while edges that are (decided to be) in the solution graph are {\it permanent}.
An induced path of length two is called a {\it conflict triple}, which is 
so named because it can never be part of a solution graph. 
To {\it cliquify} a set $S$ of vertices is to transform $G[S]$ into 
a clique by adding edges. 

A clique is permanent if each of its edges is permanent.
To {\it join} a vertex $v$ to a clique $C$ is to add all edges between $v$ 
and vertices of $C$ that are not in $N(v)$. This operation makes sense only 
when $C$ is permanent or when turning $C$ to a permanent clique.
If $v$ already contains $C$ in its neighborhood, then {\em joining}
$v$ to $C$ is equivalent to making $C\cup \{v\}$ a permanent clique.
To {\em detach} $v$ from $C$ is the opposite operation 
(of deleting all edges between $v$ and the vertices of $C$).

The first, and simplest, algorithm for Cluster Editing finds a conflict 
triple in the input graph and ``resolves'' it by exploring the
three cases corresponding to deleting one of the two edges in the path 
or inserting the missing edge. 
In each of the three cases, the algorithm proceeds recursively.
As such, the said algorithm runs in $O(3^k)$ time 
(3 cases per conflict triple). 
The same idea has been used in almost all subsequent algorithms, which 
added more sophisticated branching rules. 

We first study the $(a,d)$-Cluster Editing problem, which corresponds to
the case where $k$ is not a parameter.
This version is similar to the one introduced in \cite{KU12}
where a bound $c$ is placed on the total number 
of edge-edit operations per single vertex. The corresponding
problem is called $c$-Cluster Editing. 
When $c \geq 4$, $c$-Cluster Editing is $\np$-hard (shown also in 
\cite{KU12}). 
This does not imply, however, that $(a,d)$-Cluster Editing is 
$\np$-hard when $a+d=4$.
To see this, note that $(a,0)$-Cluster Editing is solvable in polynomial
time for any $a \geq 0$: any solution must add all the needed
edges to get rid of 
all conflict-triples, if this is possible with at most $a$ edge 
additions per vertex.

Observe that if $(a,d)$-Cluster Editing is $\np$-hard then so is 
$(a',d')$-Cluster Editing for all $a' \geq a$ and $d' \geq d$. This follows 
immediately from the definition since every instance of the first is an 
instance of the second.
We shall prove that $(2,1)$-Cluster Editing is $\np$-hard, which yields the 
$\np$-hardness of $(a,1)$-Cluster Editing for all $a > 1$.
It was shown in \cite{KU12} that $(0,2)$-Cluster Editing is $\np$-hard. We 
observe in Section 5 
that $(0,1)$-Cluster Editing is solvable in polynomial time, 
and we conjecture that $(1,1)$-Cluster Editing is solvable in 
polynomial-time. 

A kernelization algorithm with respect to an input parameter $k$
is a polynomial-time reduction procedure that yields an equivalent problem
instance whose size is bounded by a function of $k$. 
Known kernelization algorithms for Cluster Editing have so far obtained 
kernels whose order (i.e., number of vertices) is bounded by a linear 
function of $k$ \cite{Cao-Chen,Chen-Meng,Guo09}.
The most recent order-bound is $2k$ \cite{Chen-Meng}. 

We introduce the minimum acceptable cluster size, $s$, as another parameter.
This is especially useful when the input graph is preprocessed so it is 
not expected to contain outlier vertices.  
Observe that the size of a cluster is expected to be relatively
large in some correlation clustering applications, such as social 
networks \cite{cs_1}.
We shall present (in Section 4)
a simple reduction procedure that leads to solving
Cluster Editing in polynomial 
time when $s > 2(a+d)$. When $s \leq 2(a+d)$, the same reduction
procedure delivers kernels whose number of edges is bounded by 
$\frac{5k}{4}max(a,2d)(a+3d)$.
In other words, when the constraints $a$ and $d$ are 
small constants, the kernel size is linear in $k$.

\section{The $(a,d)$-Cluster Editing Problem}

We consider the case where the Cluster Editing problem is 
parameterized by 
the add and delete capacities $a$ and $d$, only.
In other words, the main objective is to check whether
it is possible to obtain a cluster graph by performing at most $a$ 
additions and $d$
deletions per vertex. We denote the corresponding problem by
$(a,d)$-Cluster Editing.

It was shown in \cite{KU12} that Cluster Editing is $\np$-hard when 
the total number of edge-edit operations per vertex is four or more. 
However, as observed earlier, the result (and proof) of \cite{KU12}
cannot be used in studying the complexity of each of the separate cases 
(for $a$ and $d$) when $a+d=4$. 
We prove in this section that $(2,1)$-Cluster Editing is $\np$-hard
by reduction
from the {\sf 4-Bounded-Positive-One-in-Three-SAT} problem, 
which is formally defined as follows:

\vspace{10pt}
\noindent
\underline{\bf 4-Bounded-positive-1-in-3-SAT}

\noindent
{\bf Given:} 
A 3-CNF formula $\phi$ in which each variable appears only
positively in at most four clauses.\\
{\bf Question:} Is there a truth assignment that satisfies $\phi$ so 
that only one variable per clause is set to true?

\vspace{10pt}

4-Bounded-positive-1-in-3-SAT was shown to be $\np$-hard 
in \cite{DF09}.
The reduction to $(2,1)$-Cluster Editing proceeds by constructing a graph
with three types of vertices: {\em variable}, {\em clause} and 
{\em auxiliary}.
Each clause $c=(x\lor y\lor z)$ is represented by a clause gadget as
shown in Figure \ref{clauseGadget}.

\vspace{15pt}

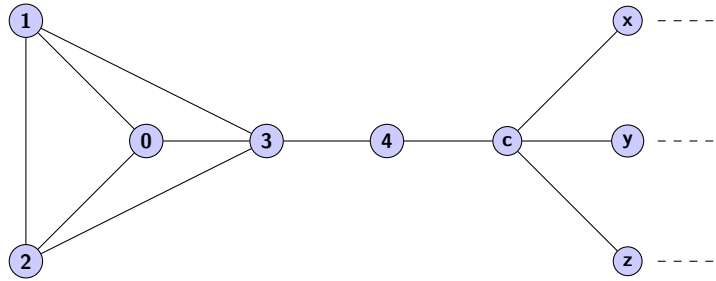
\begin{figure}[htb!]

\centering
\begin{tikzpicture}[->,
scale=0.8, every node/.style={anchor=center, scale=0.55},node distance=3cm,main node/.style={circle,fill=blue!20,draw,font=\sffamily\Large\bfseries}]

\node[main node] (0) at (3,0) {0};
\node [main node] (1) at (1,2) {1};
\node [main node] (2) at (1,-2) {2};
\node [main node] (3) at (5,0) {3};
\node [main node] (5) at (9,0) {c};
\node [main node] (7) at (11,0) {y};
\node [main node] (4) at (7,0) {4};
\node [main node] (6) at (11,-2) {z};
\node [main node] (8) at (11,2) {x};

\begin{scope}[-]
  \path
  (0) edge (2)
      edge (1)
      edge (3)
  (1) edge (2)
      edge (3)
  (2) edge (3)
  (3) edge (4)
  (4) edge (5)

 (5) edge (6)
  edge (7)
  edge (8);

\draw [dashed] (11.5,2) -- (12.5,2);
\draw [dashed] (11.5,0) -- (12.5,0);
\draw [dashed] (11.5,-2) -- (12.5,-2);

\end{scope}

\end{tikzpicture}

\caption{Clause gadget. The node corresponding to clause $c=(x \lor y \lor z)$ is 
forced to lose one of the edges connecting it to a variable gadget.}
\label{clauseGadget}
\end{figure}

\vspace{10pt}

Observe that edge $\{3,4\}$ must be deleted since the $K_4$ formed
by $\{0,1,2,3\}$ is permanent. It follows that edge $\{c,4\}$
cannot be deleted, so exactly one of the three edges $cx$, $cy$ and 
$cz$ must be deleted
by any feasible solution. The deletion of the edge $cx$ means that the
variable $x$ is set to true, otherwise it is false.

\begin{figure}[htb!]

\centering
\begin{tikzpicture}[->,
scale=0.8, every node/.style={anchor=center, scale=0.55},node distance=3cm,main node/.style={circle,fill=blue!20,draw,font=\sffamily\Large\bfseries}]

\node [main node] (0) at (3,0) {0};
\node [main node] (1) at (1,2) {1};
\node [main node] (2) at (1,-2) {2};
\node [main node] (3) at (5,0) {3};
\node [main node] (5) at (9,0) {c};
\node [main node] (7) at (11,2) {x};
\node [main node] (4) at (7,0) {4};
\node [main node] (8) at (11,-2) {y};

\begin{scope}[-]
  \path
  (0) edge (2)
      edge (1)
      edge (3)
  (1) edge (2)
      edge (3)
  (2) edge (3)
  (4) edge (5)
      edge (7)
      edge (8)
  (5) edge (8)
      edge (7)
  (7) edge (8);

\end{scope}

\end{tikzpicture}

\caption{Clause gadget resolved after setting $z$ to true while $x$ and
$y$ are set to false.}
\label{clauseGadgetResolved}
\end{figure}
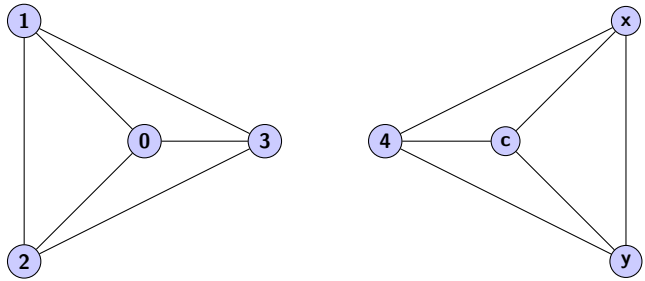

\vspace{15pt}

In the variable
gadget, shown in Figure \ref{variableGadget} below, the four edges connecting
the vertices labeled $x$ to the cycle are subject to the same edit operation:
either all four are deleted or all four are kept as permanent. 
To see this, note that it is not possible to delete exactly one edge of the 
cycle $\{1,2,3,4\}$. We either have to turn this cycle to a $K_4$ by
two edge additions and by removing each of the edges incident on vertices 
labeled $x$, or two opposite edges are deleted as shown in 
Figure \ref{variableGadgetResolved}.

\vspace{10pt}

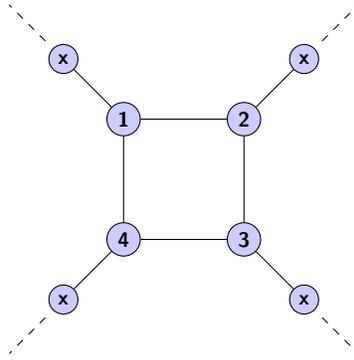
\begin{figure}[htb!]

\centering
\begin{tikzpicture}[->,
scale=0.8, every node/.style={anchor=center, scale=0.55},node distance=3cm,main node/.style={circle,fill=blue!20,draw,font=\sffamily\Large\bfseries}]

\node [main node] (0) at (0,2) {x};
\node [main node] (1) at (4,2) {x};
\node [main node] (2) at (1,1) {1};
\node [main node] (3) at (3,1) {2};
\node [main node] (4) at (3,-1) {3};
\node [main node] (5) at (1,-1) {4};
\node [main node] (6) at (4,-2) {x};
\node [main node] (7) at (0,-2) {x};

\begin{scope}[-]
  \path
  (0) edge (2)
  (2) edge (3)
      edge (5)
  (3) edge (4)
      edge (1)
  (4) edge (5)
      edge (6)
  (5) edge (7);

\draw [dashed] (-0.3,2.3) -- (-0.9,2.9);
\draw [dashed] (-0.3,-2.3) -- (-0.9,-2.9);
\draw [dashed] (4.3,2.3) -- (4.9,2.9);
\draw [dashed] (4.3,-2.3) -- (4.9,-2.9);

\end{scope}

\end{tikzpicture}

\caption{Variable gadget. Variable $x$ may belong to up to four clauses.
Each node labeled $x$ is connected to its corresponding clause gadget.}
\label{variableGadget}
\end{figure}

\vspace{15pt}

If a variable $x$ belongs to $i$ different clause(s) where $i < 4$, then 
$4-i$ vertice(s) labeled $x$ in the corresponding gadget will be pendant
(not connected to a clause gadget).  
If there is a satisfying truth assignment of a given formula of 
$4BP$-one-in-three-3SAT, then every variable that is set to true will
have its variable gadget turned into two clusters as shown in Figure
\ref{variableGadgetResolved}, and every variable $x$ that is
set to false will have its corresponding four vertices disconnected from
the $C_4$, which is turned into a $K_4$ in its gadget. In the latter case,
each of the vertices labeled $x$ will join a (corresponding) clause cluster 
as shown in Figure \ref{clauseGadgetResolved}.

\vspace{15pt}

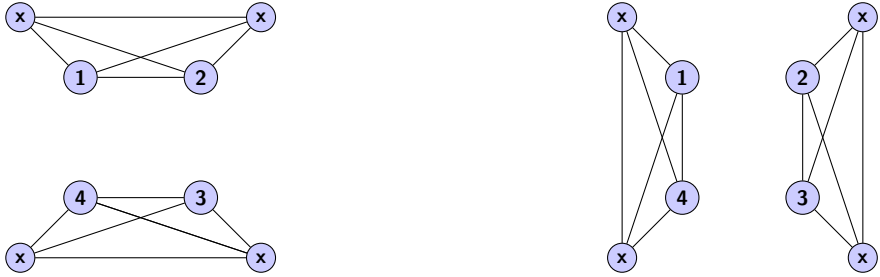
\begin{figure}[htb!]

\centering
\begin{tikzpicture}[->,
scale=0.8, every node/.style={anchor=center, scale=0.55},node distance=3cm,main node/.style={circle,fill=blue!20,draw,font=\sffamily\Large\bfseries}]

\node [main node] (0) at (0.5,2) {x};
\node [main node] (1) at (4.5,2) {x};
\node [main node] (2) at (1.5,1) {1};
\node [main node] (3) at (3.5,1) {2};
\node [main node] (4) at (3.5,-1) {3};
\node [main node] (5) at (1.5,-1) {4};
\node [main node] (6) at (4.5,-2) {x};
\node [main node] (7) at (0.5,-2) {x};

\node [main node] (8) at (10.5,2) {x};
\node [main node] (9) at (14.5,2) {x};
\node [main node] (10) at (11.5,1) {1};
\node [main node] (11) at (13.5,1) {2};
\node [main node] (12) at (13.5,-1) {3};
\node [main node] (13) at (11.5,-1) {4};
\node [main node] (14) at (14.5,-2) {x};
\node [main node] (15) at (10.5,-2) {x};

\begin{scope}[-]
  \path
  (0) edge (2)
      edge (1)
      edge (3)
  (2) edge (3)
      edge (1)
  (1) edge (3)
  (4) edge (5)
      edge (6)
      edge (7)
  (5) edge (6)
      edge (7)
      edge (6)
  (6) edge (7)

  (8) edge (10)
      edge (13)
      edge (15)
  (10) edge (13)
      edge (15)
  (13) edge (15)
  (9) edge (12)
      edge (11)
      edge (14)
  (12) edge (11)
      edge (14)
  (11) edge (14);

\end{scope}

\end{tikzpicture}

\caption{Variable gadget: the two cases where exactly two edges of $C_4$ are 
deleted.}
\label{variableGadgetResolved}
\end{figure}

Conversely, if there is a solution of the $(2,1)$-Cluster Editing instance, 
then exactly one variable-vertex neighbor of a clause-vertex is deleted. Due 
to above construction, the corresponding variable can be set to true and a 
satisfying assignment is obtained for the $4BP$-one-in-three-3SAT instance. 
We have thus proved the following.

\begin{lemma}
\label{2-1-CE}
The $(2,1)$-Cluster Editing problem is $\np$-Hard.
\end{lemma}

The membership of $(a,1)$-Cluster Editing in $\np$ is obvious.
With the above Lemma,
and our definition of $(a,d)$-Cluster Editing, we obtain the following theorem.

\begin{theorem}
\label{a-1-CE}
For $a>1$, the $(a,1)$-Cluster Editing problem is $\np$-Complete.
\end{theorem}

\section{A Reduction Procedure}

In general, a problem-reduction procedure is based on reduction rules, each 
of the form $\langle condition, action \rangle$, where $action$ is an 
operation that can be performed to obtain an equivalent instance of the 
problem whenever $condition$ holds. If a reduction is not possible, or the 
$action$ violates a problem-specific constraint, then we have a no instance. 
Moreover, a reduction rule is sound if its action results in an equivalent 
instance.

In what follows, we assume an instance $(G,k)$ of $(a,d,s,k)$-Cluster 
Editing is given. In other words, the problem is parameterized by
the add and delete capacities, as in the previous section, along with
the lower-bound on cluster size $s$ and the total number of 
edge-edit operations $k$.
Any added edge, in what follows, is automatically set to permanent.

The main reduction rules are given below. They are assumed to be applied 
successively in such a way that a rule is not applied, or re-applied, until
all the previous rules have been applied exhaustively.
Some of the rules are folklore.
We shall prove the soundness of new, non-obvious, reduction rules only.

\subsection{\bf Base-case Reductions}

\vspace{5pt}

\begin{rrule}
\label{automatic}
The reduction algorithm terminates and reports a no instance, whenever any of
$k, \alpha(v)$, or $\delta(v)$ is negative for some vertex $v \in V(G)$.
\end{rrule}

\begin{rrule}
\label{d0}
For any vertex $v$, if $\delta(v) = 0$ (or becomes zero), then $N(v)$ is 
cliquified. 
\end{rrule}

\noindent
Note that applying Rule \ref{d0} may yield negative parameters, which triggers 
Rule \ref{automatic} and causes the algorithm to terminate with a negative 
answer.

\begin{rrule}
\label{forbid} 
If $\alpha(u) = 0$, then set every non-edge of $u$ to forbidden.
\end{rrule}

\subsection{\bf Reductions Based on Conflict-Triples}

\vspace{5pt}

\begin{rrule}\label{permanent-triple} 
If $uv$ and $uw$ are permanent edges and $vw$ is a non-edge, then 
add $vw$ and decrement each of $k$, $\alpha(v)$ and $\alpha(w)$ by one. If $vw$ is a non-permanent edge,
then set $vw$ as permanent.
\end{rrule}

\begin{rrule}\label{permanent-forbidden} If $uv$ is a permanent edge and $uw$ 
is a forbidden edge, then set $vw$ as forbidden. If $vw$ exists, delete it 
and decrement $k$, $\delta(v)$ and $\delta(w)$ by one.
\end{rrule}

\subsection{\bf Reductions Based on Common Neighbors}

\vspace{5pt}

\begin{rrule}\label{common-upper}
If two non-adjacent vertices $u$ and $v$ have more than 
$2d$ common neighbors (or more than
$\delta(u) + \delta(v)$ common neighbors),
then add edge $uv$ and decrement each of $k$, $\alpha(u)$ and $\alpha(v)$ by 
one.
\end{rrule}

\noindent
{\bf Soundness:} If $u$ and $v$ are not in the same clique in the solution 
graph, then at least one of them has to lose more than $d$ edges, which 
is not possible.

\begin{rrule}\label{common-upper-adjacent} If two adjacent vertices, $u$ and 
$v$, have at least $2d-1$ common neighbors then set $uv$ as permanent edge.
\end{rrule}

\noindent
{\bf Soundness:} If the two vertices are in
different clusters of a solution graph, then deleting edge $uv$
reduces both $\delta(u)$ and $\delta(v)$ to at most $d-1$ each. Since they
have to lose their common neighbors, at least one of them
has to lose $d$ edges, which is not possible.

\begin{rrule}\label{common-lower}
If two vertices $u$ and $v$ are such that $|N(u)\setminus N(v)| > a+d$ then 
set edge $uv$ as forbidden. If $u$ and $v$ are adjacent, then delete $uv$ and 
decrement each of $k$, $\delta(u)$ and $\delta(v)$ by one.
\end{rrule}

\noindent
{\bf Soundness:} For $u$ and $v$ to be in the same cluster, at most $d$ 
neighbors may be deleted from $N(u)$ and at most $a$ neighbors can be 
added to $N(v)$.

\subsection{\bf Reductions Based on Cluster-Size}

\vspace{5pt}

\begin{rrule}\label{degree-lower} If there is a vertex $v$ 
satisfying: $s-1 > degree(v) + \alpha(v)$, then return No.
\end{rrule}

\noindent
{\bf Soundness:} Obviously, $v$ needs more than $\alpha(v)$ edges to be 
a member of a cluster in a solution graph.

\begin{rrule}\label{degree-lower-deletion} If there is a vertex 
$v$ satisfying: $\delta(v) > degree(v)+\alpha(v) -(s-1)$
set $\delta(v) = degree(v) +\alpha(v) - (s-1)$.
\end{rrule}

\noindent
{\bf Soundness:} If $\delta(v)$ edges incident on $v$ are deleted
(so $degree(v)$ is decremented by $\delta(v)$), we get 
a no-instance by Rule \ref{degree-lower}.

\begin{rrule}\label{cluster-common-lower} If $s > 2$ and two non-adjacent 
vertices $u$ and $v$ have {\bf less than} $s - 2a$ common neighbors
(or $s - \alpha(u) - \alpha(v)$ such neighbors) then set edge 
$uv$ as forbidden.
\end{rrule}

\noindent
{\bf Soundness:} If $\alpha(u)\alpha(v) = 0$,
then $uv$ is already forbidden by Rule 3. 
For $u$ and $v$ to belong to the same cluster, they must 
have at least $s-2$ common neighbors. After adding $uv$, the maximum number 
of common neighbors we can add is $2a -2$ ($a-1$ edges between $u$ and 
$N(v)$ and vice versa). The total number of common neighbors after adding 
all possible edges remains less than $s-2 (= s-2a+2a-2)$.

\begin{rrule}\label{common-lower-adjacent} If $s>2$ and two adjacent 
vertices $u$ and $v$ have $< s-2a-2$ common neighbors, then delete edge 
$uv$ and decrement each of $k$, $\delta(u)$ and $\delta(v)$ by one.
\end{rrule}

\noindent
{\bf Soundness:} The argument is similar to the previous case, except that 
each vertex must add at least $a$ neighbors of the other to obtain $s-2$ 
common neighbors.

\subsection{\bf Permanent and Isolated Cliques}

\vspace{5pt}

If a clique contains more than $2d$ vertices, then it is permanent
due to Rule \ref{common-upper-adjacent}. Moreover, no vertex can be joined
to an isolated permanent clique with more than $a$ vertices. As
a consequence, we obtain the following reduction rule.

\begin{rrule}\label{large-isolated-clique} If a clique $C$ is such that 
$N(C)\setminus C = \emptyset$
and $|C| > max(a,2d)$ then delete $C$.
\end{rrule}

The presence of permanent cliques can yield problem reductions 
that are not obtained by exhaustive applications of the above reduction rules. 
Note that a permanent edge is a special case of a permanent clique. 
The soundness of the following rules is obvious.

\begin{rrule}\label{join-clique} If a vertex $v$ has more than $d$ neighbors
in a permanent clique $C$, then $v$ is joined to $C$. 
\end{rrule}

\begin{rrule}\label{detach-clique}
Let $C$ be a permanent clique of size $> a$. If a vertex $v$ has less than 
$|C|-a$ neighbors in $C$, then $v$ is detached from $C$. 
\end{rrule}

An isolated clique is said to be {\em small} if its size is less than
$s$. In general, deleting isolated cliques is a sound reduction rule for 
the single-parameter Cluster Editing problem. In multi-parameterized
versions, we either add a parameter that bounds the number of
small isolated cliques, including outliers, 
or (to adhere to our problem formulation)
we must keep a number of such cliques. To see this, note the example
of a single isolated vertex $v$ and an isolated clique $C$ with
less than $\alpha(v)$ vertices. In this case, $v$ can be joined to 
$C$ to avoid having a small cluster that can potentially yield a
no answer.

At this stage, if there is an isolated clique $C$ in the so-far reduced
instance, then $|C| \leq max(a,2d)$. If $C$ is small (i.e.,
$1 \leq |C| < s$) then each vertex
of $C$ must be affected by at least one edge-editing operation. 
Consequently:

\begin{rrule}\label{small-isolated-cliques}
If the total number of vertices in small isolated cliques is $> 2k$ then
(halt and) report a no instance.
\end{rrule}

On the other hand, if an isolated clique $C$ is not small then it must
satisfy $s\leq |C| \leq max(a,2d)$. As observed above, 
we need to keep a few of these cliques. Since this is needed only when 
$s\geq 2$, we
can safely keep at most $k/2$ such cliques. 

\begin{rrule}\label{weak-isolated-cliques-rule} 
If $s> 1$ and there are more than $k$ non-small isolated cliques then delete 
all but at most $k/2$ of them.
\end{rrule}

\section{The Complexity of Multi-parameterized Cluster Editing}

An instance $(G,a,d,s,k)$ of Cluster Editing is said to be
reduced if the above reduction rules have been exhaustively applied to 
the input graph $G$.

Our second main theorem, proved below, addresses the optimization 
version of Cluster Editing, 
which seeks a minimum number of edge-edit operations. So $k$ is not a 
parameter in this case, but we keep the other constraints.

\begin{theorem} 
\label{poly-time}
When $s > 2(a+d)$, and $ad>0$, the Minimum Cluster Editing problem
is solvable in polynomial time.
\end{theorem}

\begin{proof}
Let $u$ and $v$ be non-adjacent vertices such that $uv$ is not forbidden. 
By Rules \ref{common-upper} and \ref{cluster-common-lower}:
$s-2a \leq |N(u)\cap N(v)| \leq 2d$. This is impossible
since $s-2a > 2d$, so $uv$ must be forbidden and
any two non-adjacent vertices must belong to different clusters. 
If $u$ and $v$ are adjacent vertices such that $uv$ is not
permanent, then by Rules \ref{common-upper-adjacent} and 
\ref{common-lower-adjacent}:
$s-2a-2 \leq |N(u)\cap N(v)| < 2d-1$. Again, this is impossible
(it implies $s < 2(a+d)+1$), so edges between adjacent vertices must
be permanent. It follows that in a reduced instance
any edge is either permanent or deleted. 
\end{proof}

In a typical clustering application, the total number of errors per 
data element is expected to be small and should be much smaller than a 
cluster size. In the seemingly common case where the cluster size is 
large compared to such error, Theorem \ref{poly-time} asserts that Cluster 
Editing is solvable (exactly) in polynomial time. 

When $s \leq 2(a+d)$, the Minimum Cluster Editing problem remains 
$\np$-hard even if $a$ and $d$ are small constants and the size of a 
cluster is not important (i.e., $s=1$).
In this case, the reduction procedure may still help to obtain faster 
parameterized algorithms. 
In fact, we shall prove that
applying the above reduction rules yields equivalent instances 
whose (total) size is bounded by a linear function of the main parameter $k$. 
The following key lemma follows from the reduction procedure.

\begin{lemma}
\label{degree-bound} 
Let $(G,a,d,s,k)$ be a reduced yes-instance of Cluster Editing. Then 
every vertex of $G$ has at most $a+3d$ neighbors.
\end{lemma}

\begin{proof}
Assume there is a vertex $v$ such that $|N(v)| > a+3d$.
By Rule \ref{common-upper}, any vertex $u$ is either a neighbor of $v$ or 
has at most $2d$ common neighbors with $v$. In the 
latter case, $v$ has more than $a+d$ vertices that are not common with $u$. 
Edge $uv$ would then be forbidden by Rule \ref{common-lower}. By 
Rules \ref{common-upper-adjacent}
and \ref{common-lower}, every edge incident on $v$ is either deleted or
becomes permanent. Applying Rules \ref{permanent-triple} 
and \ref{permanent-forbidden} exhaustively leads to cliquifying and 
isolating $N[v]$, which then results in deleting $N[v]$ due to Rule
\ref{large-isolated-clique}.
\end{proof}

\subsection{\bf The case where $a$ and $d$ are small fixed constants}

\vspace{5pt}

\noindent
It was shown in \cite{KU12} that Cluster Editing is $\np$-hard
when $a=0$ and $d=2$. This implies the $\np$-hardness of $(a,2)$-Cluster 
Editing for $a\geq 0$. 

When $d=1$, our reduction procedure
results in a triangle-free instance or a no answer. To see this observe
that every clique of size three or more becomes permanent ($2d+1=3$).
Moreover, a vertex with more than one 
neighbor in a clique of size three must be in the clique (otherwise we 
have a no instance) while an edge joining a vertex to only one member
of a triangle must be deleted. It follows that cliques of size three
or more become isolated and deleted.  

We also observe that $(0,1)$-Cluster Editing is solvable in polynomial time. 
To see this, note that a vertex of degree three must be part of a 
clique of size at least three, since at most one of its incident edges
can be deleted. Unless we have a no instance, such vertex cannot exist in a 
reduced instance, being triangle-free.
It follows that a reduced yes instance must have a maximum degree of two.
In this case, the problem is equivalent to the {\sc Maximum Matching}
problem.

At this stage, the only remaining open problem is 
whether $(1,1)$-Cluster Editing is solvable in polynomial time. 
We believe it is the case, especially since every instance is
triangle-free, as observed above. Moreover, the reader would easily 
observe that every such instance is of maximum degree three.

\subsection{\bf Kernelization}

\vspace{5pt}

\noindent
We now give a bound on the number of vertices in a reduced
instance.

\begin{theorem}
\label{kernel-order}
There is a polynomial-time reduction algorithm that takes an arbitrary 
instance of (multi-parameterized) Cluster Editing and either determines that 
no solution exists or produces an equivalent instance whose order is 
bounded by $\frac{5k}{2}max(a,2d)$.
\end{theorem}

\begin{proof}
Let $(G,a,d,k)$ be a reduced instance of Cluster Editing.
Let $A$ be the set of vertices of $G$ that are incident to edges that
must be deleted or to non-edges that must be added to obtain some 
minimum solution, if any. In other words, $A$ is the set of vertices
{\em affected} by edge-editing operations.
If $(G,a,d,k)$ is a yes instance, then $|A|\leq 2k$.
Let $B=N(A)$ and $C= V(G)\setminus (A\cup B)$. 

The connected components of $B$ are cliques, being non-affected by
edge-editing. For the same reason, every vertex $x\in A$ satisfies: 
$G[B\cap N(x)]$ is a clique. 
Let $B_x = N(x)\cap B$ for some $x\in A$
and let $N(B_x) = \{a\in A: B_x \subset N(a)\}$. 
If $|B_x \cup \{x\}| > max(a,2d)$ then our reduction procedure 
sets $B_x \cup \{x\}$ to permanent and automatically joins every element 
of $N(B_x)$ to a larger isolated clique containing $B_x$, which is
then deleted by Rule \ref{large-isolated-clique}.
Therefore $|A\cup B|\leq 2kmax(a,2d)$.

Observe that isolated
cliques of size $< s$ must be totally contained in $A$, so their
elements are part of the $2kmax(a,2d)$ vertices of $A\cup B$.
It follows by Rule
\ref{weak-isolated-cliques-rule} that $C$ 
consists of at most $k/2$ (non-small) isolated cliques, each of size 
$max(a,2d)$. 
The proof is now complete.

\end{proof}

\noindent
The following corollary follows easily from Theorem \ref{kernel-order} and
Lemma \ref{degree-bound}.

\begin{corollary}
\label{kernel-size}
There is a polynomial-time reduction algorithm that takes an arbitrary
instance of Multi-parameterized Cluster Editing
and either determines that no solution exists or produces an equivalent 
instance whose size is bounded by $\frac{5k}{4}max(a,2d)(a+3d)$.
\end{corollary}

The above kernel bound is of interest when $a$ and $d$ are 
fixed small constants. It could be particularly useful since
previously known kernelization 
algorithms for single-parameter Cluster Editing are not applicable to 
the multi-parameterized version especially due to their 
use of edge contraction, which also results in assigning 
weights to vertices.
The difficulty of applying edge contraction stems from the fact that
merging two adjacent vertices does not yield valid edge editing bounds
on the resulting vertex.
To explicate, assume we contract edge $uv$ and assign  
$\delta(u)+\delta(v)$ as edge-deletion bound to the resulting vertex.
Then we might either exceed the $d$-bound or allow the deletion of 
more than $\delta(u)$ neighbors of $u$ (from the new combined neighborhood), 
thereby obtaining a wrong solution.
 
\section{Acknowledgement}

The work on this project has been partially supported by the Lebanese American
University under grant SRDC-t2013-45.

\section{Concluding Remarks}

We considered the $(a,d)$-Cluster Editing problem, a constrained 
version of Cluster Editing where at most $a$ edges can be added
and at most $d$ edges can be deleted per single vertex.
We proved that $(a,d)$-Cluster Editing is $\np$-hard, in general, for 
any $a\geq 2$ and $d\geq 1$. 
We also observed that $(0,1)$-Cluster 
Editing is solvable in polynomial time
while the $(0,2)$ case is $\np$-hard \cite{KU12}.
It remains open whether $(1,1)$-Cluster Editing can be solved in polynomial 
time. 

We presented a reduction procedure that solves the 
Cluster Editing problem in polynomial-time when the 
smallest acceptable cluster size exceeds twice the 
total allowable edge operations per vertex.
It is worth noting that edge-editing operations 
per single data element are expected to be small compared to
the cluster size,
especially if the input is free of outliers. 
%
When the bounds on the two edge-edit operations per vertex
are small constants, and the cluster size is unconstrained,
our reduction algorithm gives a kernel whose size is linear 
in the main parameter $k$. 
Previously known kernelization algorithms 
are not applicable to the multi-parameterized version
and achieve a linear bound on the number of vertices only. 

The reduction algorithm presented in this paper has been implemented
along with the simple branching on conflict triples 
described in Section 2. 
Experiments show a promising performance 
especially on graphs obtained from clinical research 
data \cite{timothy, heidi}.  
The multi-parameterized algorithm finished consistently in 
seconds, reporting significant clusters. 

Finally we note that different Cluster Editing solutions to the same problem 
instance may differ in terms of the practical significance of obtained 
clusters. 
A possible approach would be to combine enumeration and editing,
by enumerating all possible Cluster Editing solutions. 
Moreover, in some biology applications, a data element may be an 
active member of different clusters. In such cases, the enumeration of all 
maximal cliques was used as a possible alternative \cite{abu2005}.
To permit a data element to
belong to more than one cluster, we suggest 
allowing {\em vertex-division} (AKA. vertex {\em cleaving})
as another edit operation whereby a vertex is replaced by two
different vertices. 
The number of allowed divisions per vertex can be added as (yet) another
parameter. This latter formulation is under consideration for future work.

\bibliographystyle{plain}

\bibliography{references}

\end{document}